\documentclass[10pt]{llncs}

\usepackage{latexsym}
\usepackage{amssymb}
\usepackage{epsfig}
\usepackage{amsmath}
\usepackage{url}
\usepackage{color}
\usepackage[all, knot]{xy}
\usepackage{stmaryrd}

\renewcommand{\L}{\mathcal{L}}

\newtheorem{thm}{Theorem}
\newtheorem{cor}[thm]{Corollary}
\newtheorem{lem}[thm]{Lemma}
\newtheorem{prop}[thm]{Proposition}
\newtheorem{defi}[thm]{Definition}
\newtheorem{rema}[thm]{Remark}

\newtheorem{exm}[thm]{Example}

\renewcommand{\implies}{\Rightarrow}

\newcommand{\bis}{\mathrel{\mathchoice%
{\raisebox{.3ex}{$\,
  \underline{\makebox[.7em]{$\leftrightarrow$}}\,$}}%
{\raisebox{.3ex}{$\,
  \underline{\makebox[.7em]{$\leftrightarrow$}}\,$}}%
{\raisebox{.2ex}{$\,
  \underline{\makebox[.5em]{\scriptsize$\leftrightarrow$}}\,$}}%
{\raisebox{.2ex}{$\,
  \underline{\makebox[.5em]{\scriptsize$\leftrightarrow$}}\,$}}}}

\newcommand{\la}{\langle}
\newcommand{\ra}{\rangle}

\newcommand{\mc}[1]{\mathcal{#1}}

\newcommand{\ML}{\ensuremath{\textit{ML}}}

\newcommand{\muc}{\ensuremath{L_{\mu}}}
\newcommand{\MLinf}{\ensuremath{\ML_{\infty}}}

\newcommand{\LTL}{\ensuremath{\mathit{LTL}}}
\newcommand{\CTL}{\ensuremath{\mathit{CTL}}}
\newcommand{\CTLstar}{\ensuremath{\CTL^*}}

\newcommand{\alaex}[1][]{\preccurlyeq^{#1}_e}
\newcommand{\saladi}[1][]{\prec^{#1}_d}
\newcommand{\salaex}[1][]{\prec^{#1}_e}
\newcommand{\aladi}[1][]{\preccurlyeq^{#1}_d}

\newcommand{\eqdi}[1][]{\approx^{#1}_{d}}
\newcommand{\eqex}[1][]{\approx^{#1}_{e}}
\newcommand{\incdi}[1][]{\Join^{#1}_{d}}
\newcommand{\incex}[1][]{\Join^{#1}_{e}}

\newcommand{\sepf}[3]{#1\parallel_{#2}#3}
\newcommand{\sepL}[3]{#1\parallel_{#2}#3}

\newcommand{\M}{\ensuremath{\mc{M}}}

\renewcommand{\iff}{\Leftrightarrow}

\newcommand{\ue}{\mathfrak{ue}}
\newcommand{\lftA}[1]{{#1}_{A}^{\mc{C}}}
\newcommand{\lftE}[1]{{#1}_{E}^{\mc{C}}}
\newcommand{\Ci}{\mc{C}_1}
\newcommand{\Cii}{\mc{C}_2}
\newcommand{\Linf}{\ensuremath{\L^{\infty}_{\infty}}}
\newcommand{\aladikk}{\preccurlyeq^{\kappa}_{d}}

\newcommand{\aladicc}{\preccurlyeq^{\mc{C}}_{d}}
\newcommand{\equivC}[1]{\equiv^{\mc{C}}_{#1}}

\begin{document}
\title{On Expressive Power and Class Invariance}
\author{Yanjing Wang\inst{1}\and Francien Dechesne\inst{2} \thanks{The authors are supported by 
NWO project VEMPS 612.000.528.}
}
\institute{CWI, Amsterdam, The Netherlands,
{\tt y.wang@cwi.nl}\\
\and Technische Universiteit Eindhoven, The Netherlands,
{\tt f.dechesne@tue.nl}
}
\maketitle

\begin{abstract}
In computer science, various logical languages are defined to analyze properties of systems. One way to pinpoint the essential differences between those logics is to compare their expressivity in terms of distinguishing power and expressive power. In this paper, we study those two concepts by regarding the latter notion as the former lifted to classes of models. We show some general results on lifting an invariance relation on models to one on classes of models, such that when the former corresponds to the distinguishing power of a logic, the latter corresponds to its expressive power, given certain compactness requirements. In particular, we introduce the notion of \textit{class bisimulation} to capture the expressive power of modal logics. We demonstrate the application of our results by revisiting   modal definability with our new insights.
\end{abstract}

\section{Introduction}

Logical languages are formal languages that can be used to express properties of mathematical structures: {\em models}. The standard notion for comparing how much logical languages can say about a certain class of models, are distinguishing power (can a language tell the difference between two models?), expressive power (which classes of models can be defined by a formula of the language?). 

The well-known hierarchy is as follows: distinguishing power is a coarser criterion for comparing languages than expressive power, e.g.~if two languages are equally distinguishing, they are not necessarily equally expressive, but if they are equally expressive, they are equally distinguishing. We can relate the two notions by observing that expressive power can be seen in terms of distinguishing power lifted to classes of models: a class of models is definable in a logic iff the logic can distinguish that class and its complement.

An important concept closely related to the distinguishing power in this context is \textit{structural equivalence} on models, e.g. isomorphism and bisimulation. Classic results showed that certain logics can only distinguish models up to certain structural equivalences, thus giving \textit{upper bounds} of the distinguishing power of the logics. In that case, the  structural equivalence is an \textit{invariance} relation for such logic. Under some restrictions, a structural equivalence may capture the distinguishing power of the logic precisely, such that two models can be distinguished by the logic iff they are equivalent according to the structural equivalence.

One goal of the paper is to sharpen the understanding of the difference between expressive and distinguishing power. 
We observe that expressive power can be seen as distinguishing power generalized to classes of models. This gives rise to the question whether a corresponding notion of invariance exists. In search of such notion, we generalize structural equivalence on models to structural relations on \textit{classes of models}, and use that for a notion of \textit{class invariance} for a logic that gives information about its expressive power. 

In Section~\ref{DandE} we study the expressive power as generalized distinguishing power in detail. We give some general invariance results in Section~\ref{secGen} on lifting structural equivalence on models to a relation on classes of models. We show that the two natural notions of class invariance correspond precisely to the class-equivalence and class-indistinguishability relations respectively, when restricted to \textit{compact classes}. This result is closely related to the use of compactness in existing Lindstr{\"o}m-type characterization theorems. 
In Section~\ref{secbis} we study class-bisimulation in the light of existing results on bisimulation.

We characterize modal definability by our notion of class-bisimilarity.

\paragraph{Related work.}
Comparing expressive power of different logics has always been a central issue in the study of logics. In the context of computer science, the work on comparing LTL and CTL is a notable example (a.o.~\cite{EmersonH83,BrowneClarkeGr87,ExCTLLTL,NainV07}), besides fundamental characterization theorems that capture the expressive power of modal logic and modal mu-calculus~\cite{Benthem,MuEx}. More recent results in terms of Lindstr{\"o}m-type characterization theorems also helped us sharpen our understanding on various fragments of FOL which are useful in computer science~\cite{BenthemMLLind07,LindTCate,LindTfragment}. More comprehensive discussions on the expressivity of modal logics can be found in~\cite{ModelModalHML}.

\section{Preliminaries} \label{DandE}
In this section, we study the gap between distinguishing power and expressive power. 
When comparing two logics $L_1$ and $L_2$, we assume a given class of models $\mc{M}$ on which $L_1$ and $L_2$ are interpreted; we use `classes of models' usually for `subclasses of $\mc{M}$'. All logics in this paper are assumed to be 2-valued.

\subsection{Expressive Power as a Generalized Distinguishing Power}

As we mentioned in the introduction, the distinguishing power of a logic is its power to tell two models apart. We call a formula $\varphi$ in a language $L$ a \textit{separating formula} for models $M_1$ and $M_2$ ($\sepf{M_1}{\varphi}{M_2}$), if either $M_1\vDash\varphi$ and $M_2\not\vDash\varphi$, or vice versa. We write $\sepL{M_1}{L}{M_2}$ if there exists $\varphi\in L$ with $\sepf{M_1}{\varphi}{M_2}$.

We say $L_2$ is at least as \textit{distinguishing} as $L_1$ ($L_1\aladi[] L_2$)
iff
\[\forall M_1, M_2\in\mc{M}: \sepL{M_1}{L_1}{M_2} \textit{ implies }
\sepL{M_1}{L_2}{M_2}.\]

For expressive power, we say $L_2$ is at least as \textit{expressive} (in defining properties i.e. classes of models)
as logic $L_1$ on \M\ ($L_1{\alaex[]} L_2$)
iff
\[\forall\varphi_1\in L_1\exists\varphi_2\in L_2\forall M\in\mc{M}: M\vDash \varphi_1 \textit{ iff } M\vDash \varphi_2,\]

It is not hard to prove that $L_1\alaex[] L_2$ implies $L_1\aladi[] L_2$, but the converse may fail, as the following table, comparing several modal logics, shows.  

\begin{figure}
\begin{footnotesize}

\[
\begin{array}{|c||c|c|c|}
\hline
\textit{expr} &&  &  \\
\diagdown &L_1\incex[] L_2  & L_1\salaex[] L_2  & L_1\eqex[] L_2  \\
\textit{dist}&    & &  \\ \hline\hline
 &L_1: \textrm{\textit{PDL}}  &  &  \\
L_1\incdi[] L_2 & L_2:\ML^- & \bot & \bot \\
 & \mc{M}:\textit{all Kripke models} &  &  \\ \hline
  & L_1: \textit{LTL} & L_1:\ML &  \\
L_1\saladi[] L_2 &L_2: \textit{CTL} &L_2: \textit{\MLinf}& \bot \\
 &\mc{M}:\textit{all Kripke models} & \mc{M}:\textit{all Kripke models} &  \\ \hline
  & L_1:\textrm{\textit{PDL}} & L_1:\ML &  \\
L_1\eqdi[] L_2  & L_2:\ML^- & L_2:\MLinf & \textit{any\ }L_1=L_2 \\
 &\mc{M}:\textit{image-finite models}& \mc{M}:\textit{image-finite models} &\textit{on any }\mc{M} \\ \hline
\end{array}
\]

\end{footnotesize}
\caption{Some modal logics compared with respect to expressive and distinguishing power. The symbols $\saladi[],\eqdi[]$ and $\incdi[]$ stand for strictly less distinguishing, equally distinguishing and incomparable in distinguishing power respectively; $\salaex[],\eqex[]$ and $\incex[]$ are the respective counterparts for expressive power. In the table, $\ML$ is basic modal logic; $\ML^-$ is basic modal logic with a backward modality; $\MLinf$ is basic modal logic with infinite (arbitrarily large) conjunctions; $\textit{PDL}$ is propositional dynamic logic; \LTL\ is Linear time temporal logic and \CTL\ is computation tree logic. Image-finite Kripke models are those in which from each world only finitely many worlds are accessible. All results are standard or easy to see; for \LTL\ vs \CTL\, cf.~\cite{EmersonH83}. }
\end{figure}

We now have a closer look at the notions of expressive and distinguishing power. The key observation is to view expressive power as distinguishing power on classes of models, with standard distinguishing power the special case for the restriction to singleton classes.

Like in the case for distinguishing power, we call a formula $\varphi$ in a language a \textit{separating formula} for classes of models $\Ci,\Cii\subseteq\mc{M}$ ($\sepf{\Ci}{\varphi}{\Cii}$), if either $$\Ci\subseteq\{M\in\mathcal{M}\mid M\vDash \varphi\} \textrm{ and } \Cii\subseteq\{M\in\mathcal{M}\mid M\not\vDash \varphi\}\textrm{ or vice versa. }$$  We write $\sepL{\Ci}{L}{\Cii}$ if there exists $\varphi\in L$ with $\sepf{\Ci}{\varphi}{\Cii}$.  Note that $\sepf{\{M_1\}}{\varphi}{\{ M_2\}}\iff\sepf{M_1}{\varphi}{M_2}$ and that the existence of a separating formula for two classes of models implies that they are disjoint.

We can gradually fill up the gap between distinguishing power and expressive power, with a hierarchy of notions capturing the ability of languages to distinguish pairs of {\em classes of} models, where those classes have cardinality up to $\kappa$: $L_2$ is at least as {\em $\kappa$-distinguishing} as $L_1$($L_1\aladikk L_2$) iff $L_2$ can distinguish the same classes of size up to $\kappa$ as $L_1$. Formally, $L_1\aladikk L_2 \iff$\\
$$\forall \Ci, \Cii\subseteq\mc{M} \textrm{ {\em of cardinality less than or equal to} } \kappa:
\sepL{\Ci}{L_1}{\Cii} \implies \sepL{\Ci}{L_2}{\Cii}.$$ We say $L_2$ is at least as {\em class-distinguishing} as $L_1$($L_1\aladicc L_2$) if $L_2$ can distinguish the same classes up to arbitrary size as $L_1$.

Now it is straightforward to see that $\kappa$-distinguishability coincides with standard distinguishability for $\kappa=1$. Standard expressivity coincides with class-distinguishability or $|\M|$-distinguishability if $\M$ is not a proper class as Theorem~\ref{Expr=ClassDist} shows:

 \begin{thm}\label{Expr=ClassDist}
For logical languages $L_1,L_2$ that are interpretable on a class of models \M, the following are equivalent:
\begin{enumerate}
\item $L_2$ is at least as class-distinguishing as $L_1$ ($L_1\aladicc L_2$).
\item $L_2$ is at least as expressive as $L_1$ ($L_1\alaex L_2$).
\end{enumerate}
\end{thm}

The essence of the proof of Theorem~\ref{Expr=ClassDist} consists of the fact that
 a subclass $\mc{C}$ of $\mc{M}$ is definable by $L$ iff there exists a separating formula in $L$ for $\mc{C}$ and $\overline{\mc{C}}=\mc{M}\setminus\mc{C}$. 
(Cf.~Appendix~\ref{app-Expr=ClassDist}.)

Since $\aladi$ is $\aladikk$ for $\kappa=1$, we have:

\begin{cor}
$L_1\alaex L_2$ implies $L_1\aladi L_2$.
\end{cor}

From the above results we know that expressive power can be approximated by distinguishing power for increasing sizes of classes of models.
However, the hierarchy may collapse due to the presence in the logical languages of certain connectives (which correspond to set-theoretic operations), as the following result shows:

\begin{thm}[Collapsing theorem]
If $L_2$ contains negation and $\kappa$-ary disjunction and conjunction, then  $L_1\aladi L_2$ iff $L_1{\aladikk}L_2$.
\end{thm}

\begin{proof}
Right to left is obvious, we prove from left to right.
Assume $L_1\aladi L_2$, and suppose $\Ci,\Cii\subseteq{\mc{M}}$ are of cardinality $\leq\kappa$, and $\sepL{\Ci}{L_1}{\Cii}$. Then, in particular, for each $M_1\in \Ci$ and $M_2\in\Cii$:
$\sepL{M_1}{L_1}{M_2}$. But because $L_1\aladi L_2$, then also $\sepL{M_1}{L_2}{M_2}$. For each pair $M_1\in \Ci$ and $M_2\in\Cii$, pick a formula $\psi({M_1,M_2})\in L_2$
such that $\sepL{M_1}{\psi({M_1,M_2})}{M_2}$. Without loss of generality, because $L_2$ contains negation, we can take the $\psi({M_1,M_2})$ such that they are true on the $M_1 (\in \Ci)$, and false on the $M_2 (\in\Cii)$.
But then for each $M_1\in \Ci$: $\bigwedge_{M_2\in\Cii}\psi(M_1,M_2)\in L_2$ separates $\{M_1\}$ and ${\Cii}$,
and hence: $\bigvee_{M_1\in\Ci}\bigwedge_{M_2\in\Cii}\psi(M_1,M_2)\in L_2$ separates $\Ci$ and ${\Cii}$.
So: $\sepL{\Ci}{L_2}{\Cii}$.

Hence: $L_1\aladi L_2$ implies $L_1{\aladikk}L_2$.
\qed\end{proof}
From the above theorem, it is straightforward to see that distinguishing power and expressive power coincide if we restrict to finite model classes and a logic with the usual connectives:
\begin{cor}\label{1=n}
If $L_2$ contains $\lor,\land$ and $\neg$, then for each $\kappa<\aleph_0$, $L_1\aladi L_2\iff L_1{\aladikk} L_2$ .
\end{cor}

\subsection{Induced Equivalence Relation}
In the previous subsection, we viewed distinguishing power and expressive power in terms of the abilities of the languages to separate model classes. Now we take a dual perspective, namely the equivalence induced by the limit of distinguishing power over classes, since this will lead us to the notion of invariance.

Let $\equiv_L$ be the \textit{induced equivalence relation} of language $L$ such that for all $M_1,M_2\in\M$:
$M_1\equiv_L M_2$ iff $M_1$ and $M_2$ satisfy exactly the same set of $L$-formulas. Note that $\sepL{M_1}{L}{M_2}$ iff $M_1\not\equiv_L M_2$, so we can easily prove:
\begin{thm}\label{IndEquiv=Indist}
 $L_1\aladi L_2$ iff~~$\equiv_{L_2}\subseteq \equiv_{L_1} $ .
\end{thm}

The theorem states: logical languages that induce a finer equivalence relation, have stronger distinguishing power and vice versa. We thus can compare distinguishing power of two logics by comparing their induced equivalence relations. As we mentioned in the introduction, given a logic, we would like to find a {\em structural} characterizations of its induced equivalence, as a structural measure of its distinguishing power. We call an equivalence relation $\sim$ on a class of models $\mc{M}$ an \textit{invariance} for $L$ on $\mc{M}$ if for all $M_1,M_2\in\mc{M}$:
 $$M_1\sim M_2 \Rightarrow M_1\equiv_L M_2.$$

Various structural invariance relations for logics are known, like bisimulation for modal logic and isomorphism for first-order logic. Sometimes we have a precise structural characterization of the induced equivalence, e.g. two pointed Kripke models are bisimilar iff they satisfy the same set of infinitary modal formulas \cite{mlbook}. So bisimulation coincides with the induced equivalence relation for $\MLinf$. However, for many temporal logics including $\mu$-calculus $\muc$, bisimulation is an invariance relation but not the induced equivalence relation: it is possible to find two non-bisimilar models that can not be told apart by \muc~\cite{MucalcHML}. It follows that \MLinf\ is strictly more distinguishing than $\muc$. On the other hand, $\muc$ and $\MLinf$ are not comparable in terms of expressive power since $\muc$ can define well-foundedness while $\MLinf$ can not, cf.~\cite{MucalcHML}.

We now investigate whether we can define a relation on models that similarly corresponds to expressive power. A promising candidate is class equivalence:  
\begin{defi} Define the equivalence $\equivC{L}$ induced by $L$ on classes of models by
$\Ci\equivC{L} \Cii\iff(\forall{\varphi\in L}, \Ci\vDash\varphi \text{ iff\ } \Cii\vDash\varphi)$,
where $\Ci,\Cii\subseteq\mc{M}$ and $\mc{C}\vDash\varphi$ iff $\forall M\in \mc{C},M\vDash\varphi$.
\end{defi}

However, this natural notion of equivalence notion does not match class-indis\-tinguishability, as the following Proposition shows:
\begin{prop}
For all classes of models $\Ci,\Cii\subseteq\mc{M}$:
$\sepL{\Ci}{L}{\Cii}\Rightarrow \Ci\not\equivC{L}\Cii $,
but the converse does not hold in general.
\end{prop}
\begin{proof}
$\Rightarrow$ is trivial.

To see that the converse does not hold, we may take some formula $\psi\in L$ that is contingent on $\mc{M}$ (i.e. $\psi$ can be satisfiable but not valid on $\M$). Now let $\Ci$ be the class of models defined by $\psi$, and let $\Cii$ be $\mc{M}$. Then $\Ci\subset\Cii$, so they are not distinguishable (distinguishability implies disjointness). But obviously $\Ci\not\equivC{L}\Cii$, because $\Ci\vDash\psi$ but {\em not} $\Cii\vDash\psi$. So, the right-to-left-implication does not hold.
\qed\end{proof}

The relation that corresponds to class-indistinguishability is the following:
\begin{defi} For classes of models $\Ci,\Cii\subseteq\mc{M}$, we define the symmetric relation $\asymp_L$ of $L$ by:
\[\Ci\asymp_L\Cii\iff\forall{\varphi\in L}:
(\Ci\vDash\varphi \Rightarrow \Cii\nvDash\neg \varphi)\land (\Cii\vDash\varphi \Rightarrow \Ci\nvDash\neg \varphi)\]

\end{defi}

It is easy to see that $\Ci\not\asymp_L\Cii\iff \sepL{\Ci}{L}{\Cii}$ (but note that $\asymp_L$ is not an equivalence relation, as it does not satisfy transitivity in general).

Thus we have:
\begin{thm}
$L_1\alaex L_2 \iff \asymp_{L_2}\subseteq\asymp_{L_1}.$ 
\end{thm}

\section{Class Invariance}\label{secGen}
In this section, we propose two ways to lift the structural equivalence to class similarity relations which can be used as candidates for the structural characterization of $\equivC{L}$ and $\asymp_L$. We show in Theorem~\ref{DimpE} that they are indeed class invariance relations if they are based on model invariance relations for the same language $L$.
\begin{defi}\label{classlift}
Given an equivalence relation $\sim$ on models, we can lift $\sim$ to a relation on classes of models in the following two ways:
\begin{itemize}
\item $\Ci\lftA{\sim}\Cii$ iff:
\begin{itemize}
\item for each $M_1\in \Ci$ there is $M_2\in \Cii$ such that $M_1\sim M_2$.
\item for each $M_2\in \Cii$ there is $M_1\in \Ci$ such that $M_1\sim M_2$.
\end{itemize}
\item $\Ci\lftE{\sim}\Cii$ iff there are two models $M_1\in \Ci$ and $M_2\in \Cii$ such that $M_1\sim M_2$.
\end{itemize}
\end{defi}
It is easy to see that $\lftE{\sim}$ and $\lftA{\sim}$ are indeed generalizations of $\sim$: $M\sim N\iff \{M\}\lftE{\sim}\{N\}\iff \{M\}\lftA{\sim}\{N\}.$
\begin{rema}
The above definitions give rise to some intuitive class comparison games that naturally lift the model comparison game for $\sim$. For example, considering $\lftE{\sim}$, define $G_E^{\sim}(\Ci,\Cii)$ for classes $\Ci,\Cii$ of models as follows: Verifier chooses two models $M,N$ from $\Ci$ and $\Cii$ respectively. Then he plays the model comparison game $\mathit{MC}^{\sim}(M,N)$ with Spoiler, if $\mathit{MC}^{\sim}(M,N)$ is defined. Verifier wins iff he wins in $\mathit{MC}^{\sim}(M,N)$. 
Depending on the type of $\mathit{MC}^{\sim}(M,N)$ (e.g. n-round, unbounded, infinitary) and the the results on $\mathit{MC}^{\sim}(M,N)$ with respect to certain language $L$, we may characterize $\asymp_{L'}$ with $L'$ being a fragment of $L$, by the Verifier's winning strategies of the game $G_E^{\sim}(\Ci,\Cii)$. 
(Space limitations prevent us from elaborating further on the game perspective in this paper.)
\end{rema}

We can show that if $\sim$ is an invariance for $L$ then the lifted relations $\lftA{\sim}$ and $\lftE{\sim}$ are invariance relations for the class equivalence or class indistinguishability respectively.
\begin{thm}\label{DimpE}
$\forall M,N: M\sim N\implies M\equiv_L N$ implies:
\begin{enumerate}
\item $\forall \Ci,\Cii:\Ci\lftA{\sim} \Cii \implies \Ci\equivC{L}\Cii$
\item  $\forall \Ci,\Cii:\Ci\lftE{\sim} \Cii \implies \Ci\asymp_{{L}}\Cii$
\end{enumerate}
\end{thm}
\begin{proof}
(1) is obvious, we only prove (2). If $\Ci\lftE{\sim}\Cii$, then there are $M\in \Ci, N\in \Cii$ such that $M\sim N$. For any formula $\phi$ such that $\Ci\vDash\phi$, we know $M\vDash\phi.$ Since $M\sim N$ implies $M\equiv_L N$, we have $N\vDash\phi$. Therefore $\Cii\not\subseteq \{M\mid M\nvDash\phi\}$.
\qed\end{proof}

As we have seen in the examples, $\forall M,N: M\sim N\iff M\equiv_L N$ does not imply $\forall \Ci,\Cii:\Ci\lftE{\sim} \Cii \iff \Ci\equivC{L}\Cii$ in general, otherwise there would be no difference between distinguishing power and expressive power. Now let ${\Linf}$ be a logical language with arbitrarily large (up to class-size) conjunctions and disjunctions. We then have the following straightforward results:
\begin{thm}\label{Linf}
$\forall M,N: M\sim N\iff M\equiv_{\Linf} N$ implies:
\begin{enumerate}
\item $\forall \Ci,\Cii:\Ci\lftA{\sim} \Cii \iff \Ci\equivC{\Linf}\Cii$
\item  $\forall \Ci,\Cii:\Ci\lftE{\sim} \Cii \iff \Ci\asymp_{{\Linf}}\Cii$
\end{enumerate}
\end{thm}
\begin{proof} We again only show (2), and due to Theorem~\ref{DimpE} we only need to show $\Rightarrow$.
Suppose $\Ci\asymp_{\Linf}\Cii$. Towards contradiction, suppose there are no models $M\in \Ci$, $N\in\Cii$ such that $M\sim N$. Then since $\forall M,N: M\sim N\iff M\equiv_L N$, we can find a differential formula $\varphi_{M,N}$ in $\Linf$ for each $M\in \Ci$ and $N\in \Cii$ such that $M\vDash \varphi_{M,N}$ but $N\nvDash\varphi_{M,N}$. Let $\psi$ be $\bigvee_M\bigwedge_{N}\varphi_{M,N}$ then $\Ci\vDash \psi$ but $\Cii\vDash \neg \psi$, contradiction.
\qed\end{proof}

We call a class of models $\mc{C}$ \textit{$L$-compact}, if for any set of $L$-formulas $\Sigma$: every finite subset of $\Sigma$ is satisfiable in $\mc{C}$ implies that $\Sigma$ itself is satisfiable in $\mc{C}$. For example, any class of finitely many mutually $L$-nonequivalent models is $L$-compact.\footnote{If not, there is an infinite set $\Sigma$ of $L$-formulas such that each finite subset of $\Sigma$ is satisfiable, but $\Sigma$ is not satisfiable as a whole. Then, for each equivalence class $[M]$ of $\equiv_L$ in $\mc{C}$, pick a $\phi_{[M]}\in\Sigma$ that does not hold on the models in $[M]$. Let $\Delta=\{\varphi_{[M]}|M\in\mc{C}\}\subseteq\Sigma$, then $\Delta$ is finite but not satisfiable on $\mc{C}$. Contradiction.} In the following we always assume that $L$ contains binary conjunction $\land$ with its normal semantics. We show that if we restrict to $L$-compact classes then $\lftA{\sim}$ and $\lftE{\sim}$ match $\equivC{L}$ and $\asymp_{L} $ exactly, given that $\sim$ matches $\equiv_L$ exactly.  To obtain this result, we use the following lemma: 

\begin{lem}\label{lem-compactC}
If $\forall M,N\in\mc{M}: M\sim N\iff M\equiv_L N$ then for any $L$-compact $\mc{C}\subseteq\M$ and $N\in\M$, (1) implies (2):
\begin{enumerate}
\item for all $\varphi\in L: \mc{C}\vDash\varphi\implies N\vDash\varphi$.
\item there is a model $M\in \mc{C}$ such that $M\sim N$.
\end{enumerate}
\end{lem}
\begin{proof}
Suppose (1), and towards contradiction suppose there is no model in $\mc{C}$ such that $M\sim N$. Let $\Sigma=\{\varphi\mid N\vDash \varphi\}.$ Since $\forall M,N\in\mc{M}: M\sim N\iff M\equiv_L N$, $\Sigma$ is not satisfiable in $\mc{C}$, otherwise there is a model $M\in\mc{C}$ such that $M\sim N$. Since $\mc{C}$ is $L$-compact, there is a finite subset $\Delta$ of $\Sigma$ such that $\Delta$ is not satisfiable in $\mc{C}$. Clearly, $\mc{C}\subseteq\{M\in\M|M\nvDash\bigwedge\Delta\}$. From (1) we know $N\nvDash\bigwedge\Delta$ contradictory to $N\vDash\Sigma$.
\qed\end{proof}
\begin{thm}\label{thm-compactC}
Suppose $\forall M,N\in\mc{M}: M\sim N\iff M\equiv_L N$. For $L$-compact classes $\Ci$ and $\Cii$ of $\M$ the following hold: 
\begin{enumerate}
\item  $ \Ci\lftA{\sim}\Cii\iff \Ci\equivC{L}\Cii$
\item $ \Ci\lftE{\sim}\Cii\iff \Ci\asymp_L \Cii$
\end{enumerate}
\end{thm}

\begin{proof}
We only show the direction $\Leftarrow$ in (2). The proof of (1) is similar but simpler.

Suppose $\Ci$ and $\Cii$ are $L$-compact and $\Ci\asymp_L \Cii$. Now suppose towards contradiction that $\Ci\lftE{\not\sim}\Cii$. Then there are no models $M\in\Ci$ and $N\in\Cii$ such that $M\sim N.$ We claim the following:
$$\textrm{ for each } N\in\Cii \textrm{ there is some }\varphi_N\in L \textrm{ such that }\Ci\vDash\varphi_N \textrm{ and }N\nvDash \varphi_N $$
Suppose not, then there would be a model $N\in\Cii$ such that for all $\varphi:\Ci\vDash\varphi_N \implies N\vDash \varphi_N$. From Lemma~\ref{lem-compactC} there is $M\in \Ci$ such that $M\sim N$, contradictory to the assumption.

Let $\Sigma$ be the collection of such formulas: $\{\varphi_N\mid N\in\Cii\}$.
Clearly, $\Sigma$ is not satisfiable in $\Cii$, but $\Ci\vDash\Sigma$. Since $\Cii$ is $L$-compact, there is a finite subset $\Delta$ of $\Sigma$ which is not satisfiable in $\Cii$.  It is easy to see that $\Cii\subseteq \{M\mid M\nvDash\bigwedge\Delta\}$. However, $\Ci\vDash\bigwedge\Delta$. Therefore $\Ci\not\asymp_L \Cii$, contradiction.
\qed\end{proof}
From the definition of $L$-compact classes, it is straightforward to see that:
\begin{prop}\label{prop-CC2CL}
Suppose $L$ contains negation. If for every class $\mc{C}\subseteq\M$ there is an $L$-compact class $\mc{C}'$ such that $\mc{C}\equivC{L}\mc{C}'$, then $L$ is compact. 
\end{prop}
\begin{proof}
Given an arbitrary set of $L$-formulas $\Sigma$, if every finite subset of $\Sigma$ is satisfiable in $\mc{C}$ then every finite subset of $\Sigma$ is satisfiable in the compact $\mc{C}'$ (otherwise there would be a finite $\Delta\subseteq\Sigma$ such that $\mc{C}\nvDash\neg\bigwedge\Delta$ but $\mc{C}'\vDash\neg\bigwedge\Delta$, which is a contradiction). Since $\mc{C}'$ is compact, there is a model $M\in \mc{C}'$ such that $M\vDash\Sigma$. Thus $L$ is compact.
\qed\end{proof}

\begin{prop}\label{prop-CdefC}
 If $L$ is compact, then $L$ only defines $L$-compact classes. 
\end{prop}
\begin{proof}
 Suppose $L$ is compact and defines the non-compact class $\mc{C}\subseteq{\mc{M}}$ by $\varphi\in L$. Let $\Sigma$ be a witness to the non-$L$-compactness of $\mc{C}$, i.e.\ every finite subset of $\Sigma$ is satisfiable on $\mc{C}$, but not $\Sigma$ itself. Note that every finite subset of $\Sigma'=\Sigma\cup\{\varphi\}$ is then satisfiable on $\mc{C}$, and therefore on $\mc{M}$. By compactness of $L$, $\Sigma'$ is then satisfied by some $M\in\mc{M}$. But then $M\vDash\Sigma$ and $M\in\mc{C}$ by $M\vDash\varphi$. Contradiction.
\qed\end{proof}

\begin{thm}\label{thm-main}
 If $L_2$ is a language extension of $L_1$  ($L_1\subseteq L_2$), then:
\begin{enumerate}
\item An $L_2$-compact class is also $L_1$-compact. 
\item If $L_1$ and $L_2$ are compact and contain $\neg$ then they have the same distinguishing power iff they have the same expressive power.
\item If $L_1$ and $L_2$ have the same distinguishing power and contain $\neg$ then: $L_1\salaex L_2\iff L_2$ defines some non-$L_2$-compact class which can not be defined by $L_1$. 
\end{enumerate}
\end{thm}
\begin{proof}
For (1): trivially by $L_1\subseteq L_2$. 

For (2), $\Leftarrow$ is obvious; we need to show $\Rightarrow.$  Suppose $L_1$ and $L_2$ have the same distinguishing power. Then $\equiv_{L_1}=\equiv_{L_2}$. Since $L_2$ extends $L_1$, we only need to show $L_2\alaex L_1$. From Proposition~\ref{prop-CdefC} and the fact that $L_1$ and $L_2$ are compact, we know that $L_1$ and $L_2$ can only define $L_1$- and $L_2$-compact classes respectively. We only need to show that if $L_2$ can define an $L_2$-compact class $\mc{C}$, then $L_1$ can also define it,i.e.\ for any $L_2$-compact class $\mc{C}: \mc{C}\not\asymp_{L_2}\overline{\mc{C}}\implies \mc{C}\not\asymp_{L_1}\overline{\mc{C}}$. Now suppose for an $L_2$-compact class $\mc{C}:$ $\mc{C}\not\asymp_{L_2}\overline{\mc{C}}$. Since $L_2$ contains negation, from Proposition~\ref{prop-CdefC} it is easy to see that $\overline{\mc{C}}$ is also $L_2$-compact. From statement (1), $\mc{C}$ and $\overline{\mc{C}}$ are also $L_1$-compact.  From Theorem~\ref{thm-compactC}, we know that it is \textit{not} the case that $\mc{C}\lftE{(\equiv_{L_2})}\overline{\mc{C}}$.\footnote{cf.~Def.\ref{classlift}, replace $\sim$ by $\equiv_{L_2}$.} Since $\equiv_{L_1}=\equiv_{L_2}$, it is \textit{not} the case that $\mc{C}{\lftE{(\equiv_{L_1})}}\overline{\mc{C}}$. Again from Theorem~\ref{thm-compactC}, $\mc{C}{\not\asymp_{L_1}}\overline{\mc{C}}$.

For (3), $\Leftarrow$ is trivial. For $\Rightarrow,$ suppose towards contradiction that $L_1\salaex L_2$, but that for all classes $\mc{C}$ that are {\em not} $L_2$-compact: if $L_2$ can define $\mc{C}$, then so can $L_1$. Because $L_1\salaex L_2$, then there must be an $L_2$-compact class $\mc{C}$ which is definable by $L_2$ but not definable by $L_1$. From Proposition~\ref{prop-CdefC}, $\mc{C}$ and $\overline{\mc{C}}$ are both $L_1$- and $L_2$-compact. However, according to Theorem~\ref{thm-compactC}, with similar technique used in the proof of (2) we can show that $\mc{C}$ is then definable by $L_1$, contradiction. 
\qed\end{proof}

\paragraph{Discussion} Statements (2) and (3) of Theorem~\ref{thm-main} are closely related to some well-known results. Recent results on Lindstr\"{o}m-type theorems for fragments $L$ of first-order logic are often formalized in the following form~\cite{BenthemMLLind07,LindTCate,LindTfragment}:
\begin{center}
\textit{A logic $L'$ that extends $L$ has equal expressive power as $L$ iff:\\
 (a) $L'$ is compact and (b) $L'$-formulas are invariant for some $\sim$. }
\end{center}
It follows from statement (2) that for compact logics, in order to obtain such theorem from right to left, we only need to show that the extended logic $L'$ has the same {\em distinguishing power} as $L$, which in principle is weaker. 
This gives an alternative perspective on the classical model theoretic arguments in~\cite{ModelModalHML,MCSF}, that use compactness for proving characterization results. Condition (b) helps to limit the distinguishing power with the presence of (a) but with some hurdles to overcome, cf.\ the discussions in~\cite{LindTfragment}.

Statement (3) gives us a hint on how to show that two equally distinguishing logics have different expressive power. This task is often considered hard. However, due to statement (3), we only need to concentrate on non-$L$-compact classes. This could explain the complicated constructions in many works on expressivity of temporal logic. For example, in~\cite{BrowneClarkeGr87}, it is shown that if a $\CTLstar$ formula does not correspond to a $\CTL$ formula then it must have an infinite number of mutually nonequivalent finite models. Now this is just a straightforward consequence of our general result.

\section{Application: Class Bisimulation and Modal Logic}\label{secbis}
In this section, we focus on \textit{class bisimulation} $\lftE{\bis}$, the lifted notion of bisimulation $\bis$. We will demonstrate that such particular structural relations on classes can be studied in a similar way as the equivalence relations on models. As we will see in the following, most of the definitions and results about bisimulation and modal logic can be adapted to the corresponding ones for class bisimulation. As an example, we show a class bisimulation characterization of modal definability similar to the previous result in \cite{UUpuremodal}.

\subsection{Standard Bisimulation}
We first recall some standard results of bisimulation and modal logic, which we will use intensively in the next subsection.
Let $M$ and $N$ be two Kripke models with set of states $S$ and $T$. A binary (symmetric) relation $R \subseteq S\times T$ is a \textit{bisimulation} relation if $sRt$ implies:
\begin{enumerate}
\item $V(s)=V(t)$
\item if $s\to s'$ in $M$ then there is a $t'$ such that $t\to t'$ in $N$ and $s'Rt'$;
\item if $t\to t'$ in $M$ then there is a $s'$ such that $s\to s'$ in $N$ and $t'Rs'$.
\end{enumerate}
For two pointed models $M,s$ and $N,t$, we say $M,s$ is bisimilar to $N,t$ ($M,s\bis N,t$) if there is a bisimulation $R$ between $M$ and $N$ and $(s,t)\in R$. Note that in this section we will only talk about \textit{pointed Kripke models}. We may also call them \textit{pointed models} or simply \textit{models} for short.

We call a pointed Kripke model $M,s$ \textit{modally-saturated} or simply \textit{m-saturated} if for any state $w\in M$ and any set of $\ML$ formula $\Sigma$: if every finite subset of $\Sigma$ is satisfiable at the successors of $w$ then $\Sigma$ is satisfiable at some successor of $w$. We say a class of models $\mc{M}$ for a logic $L$ has the \textit{Hennessy-Milner property} for $\ML$ if for all models $M,s,N,t\in\mc{M}$:$$M,s\equiv_\ML N,t \Rightarrow M,s\bis N,t.$$

Here we list some standard results cf.~\cite{mlbook}. First, bisimulation is an invariance for \ML:
\begin{thm}\label{bisinv}
If pointed Kripke models $M,s\bis N,t$ then $M,s\equiv_{\ML} N,t.$
\end{thm}
Let $\ue(M),\pi_s$ be the \textit{pointed} \textit{ultrafilter extension}\footnote{An ultrafilter extension of a Kripke model $M=\la W,R,V\ra$ is again a Kripke model where the set of worlds is the set of ultrafilters over $W$; two ultrafilters $u$ and $v$ are linked by a relation iff $X\in v$ implies the set of points which \textit{sees} $X$ through a relation in $M$ is in $u$. The valuation of proposition letter $p$ at $u$ is true iff $V(p)\in u$. For the formal definition of ultrafilter extension cf.~\cite{mlbook}.} of a model $M,s$, with the principal ultrafilter $\pi_s$ generated by $s$ as its designated point. We know that:
\begin{thm}\label{ueequiv}
For any pointed model $M,s:$ $M,s\equiv_{\ML}\ue(M),\pi_s$ and $\ue(M),\pi_s$ is m-saturated.
\label{uethem}
\end{thm}
The above Theorem implies the following \textit{bisimulation-somewhere-else} result:
\begin{thm}\label{bissomewhereelse}
For pointed models: $M,s\equiv_{\ML} N,t \iff \ue(M),\pi_s\bis \ue(N),\pi_t$.
\end{thm}
Let $\mc{C}^m$ be the class of m-saturated pointed Kripke models. Bisimulation coincides with the induced equivalence relation of \ML\ on $\mc{C}^m$:
\begin{thm}\label{biscoi}
$\mc{C}^m$ enjoys Hennessy-Milner property for \ML. Thus for any $M,s$, $N,t\in \mc{C}^m: M,s\bis N,t \iff M,s\equiv_{\ML} N,t$.
\end{thm}
The last standard result we want to mention here is that bisimulation coincides with induced equivalence relation of \MLinf\ (\ML\ with conjunction over any set):
\begin{thm}\label{MLinf}
$M,s\bis N,t \iff M,s\equiv_{\MLinf} N,t$
\end{thm}
Moreover, we can actually show that the class of $m$-saturated models is the biggest class which contains $\mc{C}^m$ and satisfies Hennessy-Milner property for \ML.
\begin{thm}\label{biggesthmclass}
For any class of pointed models (generated by the designated points) $\mc{C}\supset\mc{C}^m$, there exist models $M,s,N,t\in \mc{C}$ such that $M,s\equiv_\ML N,t$ but $M,s\not \bis N,t$.
\end{thm}
\begin{proof}
The proof makes use of the next lemma, Lemma~\ref{m-bis}. See Appendix~\ref{app-biggesthmclass}.
\qed\end{proof}
\begin{lem}\label{m-bis}
If $M,s\bis N,t$ and $N,t$ is m-saturated, then $M,s$ is m-saturated too.
\end{lem}
\begin{proof}
See Appendix~\ref{app-m-bis}.
\qed\end{proof}
\subsection{Class Bisimulation}\label{Cbis}
In this section, we prove the analogies of the results in the previous section w.r.t class bisimulation $\lftE{\bis}$. Corresponding results w.r.t $\lftA{\bis}$ can be also proved with little adaptation.

From Theorem~\ref{DimpE} and Theorem~\ref{bisinv}, we know that $\lftE{\bis}$ is the class invariance for \ML, namely for any classes of models $\Ci$ and $\Cii$: $$\Ci\lftE{\bis} \Cii\implies \Ci\asymp_{\ML} \Cii.$$ However the converse is not true even when we restrict to m-saturated models. According to Theorem~\ref{thm-compactC}, such counterexample must involve non-compact classes. For example, let $\Ci=\{M_1, M_2,\dots\}$ be the set of all the finite walks while $\Cii=\{M_0\}$ where $M_0$ is a single point with self-loop.
\begin{exm}\label{ex2}
$\Ci\lftE{\not\bis} \Cii$ but $\Ci\asymp_\ML \Cii.$\\ 
\medskip
\centerline{
$\xymatrix@R-5pt@C-5pt@M-4pt{\bullet\ar[d] & \bullet\ar[d] &\bullet\ar[d]  &\dots & & & \bullet\ar@(ul,ur) \\
\bullet & \bullet\ar[d] &\bullet\ar[d]  & \dots &  & & \\
& \bullet &\bullet\ar[d]  &\dots &  & & \\
&         &\bullet  & \dots&  & &\\
& &\Ci & & & & \Cii }$
}
\end{exm}

Note that $M_0$ is not bisimilar to any of the models in $\Ci$. However we can show that $\Ci\asymp_L \Cii$. Suppose not, then there is a formula $\varphi\in L$ such that $\Ci\vDash \neg \varphi $ and $\Cii\vDash \varphi$. Since $\varphi$ is satisfiable at the tree unraveling of $M_0$, it must be satisfiable at a finite submodel of it due to the finite depth property of \ML. However, any finite submodel of the unraveling of $M_0$ is just a finite walk. Thus there is a $k$ such that $M_k\vDash \varphi$. Therefore $\Ci\nvDash \neg \varphi$. Contradiction.

Like in the case of standard bisimulation, we want to identify the collections of model classes where class bisimulation coincides with the class equivalence relation of modal logic. Thus we need to lift the notion of m-saturated models.
\begin{defi}[Modally-saturated class]
A class $\mc{C}$ of pointed Kripke models is \textit{m-saturated} if each model in $\mc{C}$ is m-saturated and $\mc{C}$ is $\ML$-compact.
\end{defi}
The following is a corollary of Theorem~\ref{biscoi} and Theorem~\ref{thm-compactC}:
\begin{thm}\label{cbiscoin}
If $\Ci$ and $\Cii$ are m-saturated then $ \Ci\lftE{\bis}\Cii\iff \Ci\asymp_L \Cii$..
\end{thm}

The \textit{class-bisimulation-somewhere-else} result corresponding to Theorem~\ref{bissomewhereelse} requires the following generalization of ultrafilter extensions which aims to make a class m-saturated.
\begin{defi}[Ultrafilter extension of classes]
Given a class $\mc{C}$ of pointed models, let $S_{\mc{C}}$ be the set of designated points of the models in $\mc{C}$. We define the  \textit{ultrafilter extension} of a class $\mc{C} $(written $\ue(\mc{C}))$ as a new class of models as follows:  $$\ue(\mc{C})=\{\ue(\biguplus\mc{C}),u\mid u\in U_\mc{C}\}$$ 
where $U_{\mc{C}}=\{u\mid u\in\ue(\biguplus\mc{C})$ and $\forall X\in u: X\cap S_{\mc{C}}\not=\emptyset \}$, and $\biguplus\mc{C}$ the disjoint union of models in $\mc{C}$ (as unpointed models). Then $\ue(\{M,s\})=\{\ue(M),\pi_s\}$.
\end{defi}
Note that our purpose is to make a class m-saturated by a operation on class of models, like ultrafilter extension does on single model. This means we need to include many auxiliary models to make the class \ML-compact. It is not hard to see that the pointed models in $\ue(\mc{C})$ are defined differently than the \textit{Ultraunions} in \cite{UUpuremodal}, where $\ue(\biguplus\mc{C}),u$ is called an \textit{ultraunion} of $\mc{C}$ iff $u\in U_{\mc{C}}'=\{u\mid  u\in \ue(\biguplus\mc{C}) \textrm{ and $u$ contains all the cofinite subsets of $S_\mc{C}$}\}$. First of all, $U_{\mc{C}}'$ does not contain principal ultrafilters generated by $s\in S_{\mc{C}}$. Moreover, when we restrict ourselves to non-principal ultrafilters of $\biguplus\mc{C}$, $U_{\mc{C}}'$ is still a subset of $U_{\mc{C}}$. To see this we can prove:
\begin{prop}\label{prop-diff}
If $S_{\mc{C}}$ is infinite, $U_{\mc{C}}'\subseteq U_{\mc{C}}|_{np}$ where $$U_{\mc{C}}|_{np}=\{u\mid u\in U_\mc{C} \textrm{ and $u\not=\pi_s$ for any  $s\in S_\mc{C}$} \}.$$
Moreover if $\biguplus\mc{C}\backslash S_{\mc{C}}$ is infinite, then $U_{\mc{C}}'\subset U_{\mc{C}}|_{np}.$
\end{prop}
\begin{proof}
See Appendix~\ref{app-prop-diff}.
\qed\end{proof}
Now we want to show an analogy of Theorem~\ref{bissomewhereelse}:
\begin{thm}[Class Bisimulation somewhere else]\label{cbissomewhereelse} For all classes of pointed Kripke models  $\Ci$ and $\Cii$:
$\Ci\asymp_{\ML}\Cii\iff \ue(\Ci)\lftE{\bis} \ue(\Cii).$

\end{thm}
To prove this theorem, we need Theorem~\ref{cbiscoin} and the following two lemmas.

\begin{lem}\label{cueequiv}
For any class of pointed Kripke models $\mc{C}$, any $\varphi\in\ML:$
$$\mc{C}\vDash\varphi\iff \ue(\mc{C})\vDash \varphi$$
\end{lem}
\begin{proof}
$\Leftarrow:$ Note that for all $s\in S_{\mc{C}}$: $\ue(\biguplus\mc{C}),\pi_s\in \ue(\mc{C})$ where $\pi_s$ is the principal ultrafilter generated by $s$. Assume $\ue(\mc{C})\vDash \varphi$ then for all $s\in S_{\mc{C}}$: $\ue(\biguplus\mc{C}),\pi_s\vDash\varphi.$ Thus from Theorem~\ref{ueequiv} and the preservation result for disjoint union, $\mc{C}\vDash\varphi$.

$\Rightarrow:$ Suppose $\mc{C}\vDash\varphi$, then $S_\mc{C}\subseteq V_{\biguplus\mc{C}}(\varphi)$ where $V_{\biguplus\mc{C}}(\varphi)=\{w\mid  \biguplus\mc{C},w\vDash \varphi\}$. Now we need to show that for all $u\in U_{\mc{C}}$: $\ue(\biguplus\mc{C}),u\vDash \varphi.$ Since $\ue(\biguplus\mc{C})$ is the ultrafilter extension of $\biguplus\mc{C}$, it amounts to prove for every $u\in U_{\mc{C}}$: $$ V_{\biguplus\mc{C}}(\varphi)\in u $$ Suppose not, then $\overline{V_{\biguplus\mc{C}}(\varphi)}\in u$, due to the fact that $u$ is an ultrafilter. Since $S_\mc{C}\subseteq V_{\biguplus\mc{C}}(\varphi)$, $\overline{V_{\biguplus\mc{C}}(\varphi)}\cap S_{\mc{C}}=\emptyset$, in contradiction to the assumption that $u\in U_{\mc{C}}.$
\qed\end{proof}

\begin{lem}\label{cuems}
For any class $\mc{C}$ of pointed Kripke models, $\ue(\mc{C})$ is m-saturated.
\end{lem}
\begin{proof} Since all the models in $\ue(\mc{C})$ are ultrafilter extensions thus m-saturated, we only need to show that $\ue(\mc{C})$ is \ML-compact.  
Given a class $\mc{C}$ of pointed Kripke models, suppose every finite subset of $\Sigma$ is satisfiable in $\ue(\mc{C})$. We claim that:
$$ \textrm{every finite subset of $\Sigma$ is satisfiable in $\{\biguplus\mc{C},s\mid s\in S_{\mc{C}}\}$ }$$
Suppose not, then there is a finite $\Delta\subseteq \Sigma$ such that $\Delta$ is not satisfiable in $\{\biguplus\mc{C},s\mid s\in S_{\mc{C}}\}$. Thus for any $s\in S_{\mc{C}}: \biguplus\mc{C},s\vDash\neg\bigwedge\Delta$. Therefore from the preservation property of the disjoint union, for any pointed model $M,s\in \mc{C}:M,s\vDash\neg\bigwedge\Delta$ and then $\mc{C}\vDash\neg \bigwedge\Delta.$ From Lemma~\ref{cueequiv}, $\ue(\mc{C})\vDash\neg\bigwedge\Delta$ contradictory to the assumption that every finite subset of $\Sigma$ is satisfiable in $\ue(\mc{C})$.

Now let $$F_{\mc{C}}(\Sigma)=\{V'_{\biguplus\mc{C}}(\Delta)\mid \Delta \textrm{ is a finite subset of }\Sigma\}$$ where $V'_{\biguplus\mc{C}}=\{s\mid s\in S_{\mc{C}}\textrm{ and } \biguplus\mc{C},s\vDash\Delta\}.$ It is not hard to see that $F_{\mc{C}}(\Sigma)$ has the finite intersection property since:
\begin{enumerate}
\item for any finite subsets $\Delta$ and $\Delta'$ of $\Sigma:$
$V'_{\biguplus\mc{C}}(\Delta)\cap V'_{\biguplus\mc{C}}(\Delta')=V'_{\biguplus\mc{C}}(\Delta\cup\Delta')$
\item $ V'_{\biguplus\mc{C}}(\Delta\cup\Delta')\not=\emptyset \textrm{ according to the above claim}.$
\end{enumerate}
From Ultrafilter Theorem cf.~\cite{mlbook}, $F_{\mc{C}}(\Sigma)$ can be extended to an ultrafilter $u^*$. Now for any $\varphi\in\Sigma$, clearly $V_{\biguplus\mc{C}}(\varphi)\in u^*$ since $V'_{\biguplus\mc{C}}(\varphi)\subseteq V_{\biguplus\mc{C}}(\varphi)$ and $V'_{\biguplus\mc{C}}(\varphi)\in u^*$. Thus from the construction of ultrafilter extension, $\ue(\biguplus\mc{C}),u^*\vDash\varphi.$ Therefore $\ue(\biguplus\mc{C}),u^*\vDash\Sigma.$

We now only need to prove that $u^*\in U_{\mc{C}}$, namely for all $X\in u^*: X\cap S_{\mc{C}}\not=\emptyset.$ Suppose not, then there is an $X$ such that $X\cap S_{\mc{C}}=\emptyset$. However, since $F_{\mc{C}}(\Sigma)\subseteq S_{\mc{C}}$, there is an $Y\in u^*:X\cap Y=\emptyset$. Then $u^*$ is not an ultrafilter, contradiction.
\qed\end{proof}

The previous two lemmas complete the proof for Theorem~\ref{cbissomewhereelse}. 

As a corollary of Theorem~\ref{cbissomewhereelse} we can give a class-bisimulation characterization of the modally-definable classes:
\begin{cor}\label{cordef}
A class of pointed model $\mc{C}$ is \textbf{not} definable in $\ML$ iff $\ue(\mc{C})\lftE{\bis}\ue(\overline{\mc{C}}).$
\end{cor}
\begin{proof}
Follows from Theorem~\ref{cbissomewhereelse}.
Note that $\mc{C}$ is not definable by a single \ML\ formula iff $\mc{C}\asymp_{\ML}\overline{\mc{C}}$. 
\qed\end{proof}
It follows that:
\begin{cor}
A class of pointed Kripke model $\mc{C}$ is definable in $\ML$ iff $\mc{C}$ and $\overline{\mc{C}}$ are closed under bisimulation and $\ue$(namely $\ue(\mc{C})\subseteq \mc{C}$).
\end{cor}
\begin{proof}
$\Rightarrow$ is straightforward. We now prove $\Leftarrow$. Note that from Corollary~\ref{cordef}, we only need to show
$$\textrm{ $\mc{C}$ and $\overline{\mc{C}}$ are closed under bisimulation and $\ue$ implies $\ue(\mc{C})\lftE{\not\bis}\ue(\overline{\mc{C}})$}$$
Suppose $\ue(\mc{C})\lftE{\bis}\ue(\overline{\mc{C}})$, namely there are $M,s\in\ue(\mc{C})$ and $N,t\in\ue(\overline{\mc{C}})$ such that $M,s\bis N,t$. Now suppose $\mc{C}$ and $\overline{\mc{C}}$ are closed under $\ue$, then $M,s\in\mc{C}$ and $N,t\in\overline{\mc{C}}$. Since $M,s\bis N,t$ then $\mc{C}$ is not closed under bisimulation.
\qed\end{proof}

We have to be careful if we want to obtain the analogy of Theorem~\ref{MLinf}. A natural question to ask is whether we can get rid of class-size conjunction and disjunctions in Theorem~\ref{Linf}, but only use \MLinf\ that is the modal logic with conjunctions over arbitrarily large sets. Unfortunately, we observe that the following is not true:
$$\Ci\asymp_{\MLinf}\Cii\iff \Ci\bis^{E}_{\mc{C}} \Cii \quad (\divideontimes)$$

Consider an example from \cite{ModelModalHML} (where it is used to show that \MLinf\ can not define well-foundedness). For each ordinal $\alpha$, let $M_\alpha=\{\{\beta \mid \beta\leq \alpha\},R\}$ be its reverse, i.e.\ $R=\{(\beta,\beta')\mid \beta'<\beta\leq\alpha\}$. Let $M_\alpha'$ be the modification of $M_\alpha$ with $R$ replaced by $R'=R\cup\{(\alpha,\alpha)\}.$ Now let $C_1$ be the class of reserved ordinals $M_\alpha$ and $C_2$ be the class of modified reversed ordinals $M_\alpha'$. It is not hard to see that none of the two models in these two classes are bisimilar to each other.  However, from \cite{ModelModalHML} we know that there is no \MLinf\ formula to separate $\Ci$ from $\Cii$, namely $\Ci\asymp_{\MLinf} \Cii$. Thus $\divideontimes$ does not hold.

As an analogy of Theorem~\ref{biggesthmclass} about Hennessy-Milner property, we can actually show that m-saturated classes constitute the core of the maximal collection $\mc{K}$ of model classes which has the following property: $$\forall \Ci,\Cii\in \mc{K}: \Ci\bis^E_\mc{C}\Cii\iff \Ci\asymp_\ML \Cii \quad(\circledast)$$
Let $\mc{K}^m$ be the collection of model classes which extend some m-saturated classes by possibly adding modally-equivalent m-saturated models. Given a model class $\mc{C}$, let $\mc{C}|_m$ be the class of m-saturated models in $\mc{C}$. Formally, $$\mc{K}^m=\{ \mc{C}\mid \mc{C}|_m \textrm{ is m-saturated and }\forall M\in\mc{C}\exists M'\in\mc{C}|_m:M\equiv_\ML M' \}$$
We now prove the following theorem:
\begin{thm}\label{largestcollclassthm}
$\mc{K}^m$ is the largest collection of model classes that contains $\mc{K}^m$ and satisfies $(\circledast)$.
\end{thm}
\begin{proof}
See Appendix~\ref{app-largestcollclassthm}.
\qed\end{proof}

\section{Conclusion and Future Work}
In this paper, we studied the expressive power as class-distinguishing power in depth. We lifted the notion of invariance to classes of models aiming to capture the expressive power of logics. On compact classes, the lifted notions of class invariance captures the expressive power precisely. In particular, we focused on the notion of \textit{class bisimulation} and demonstrated the application of our results by revisiting modal-definability with our new insights which may shorten the proof. It makes clear that compactness of classes of models plays an important role in obtaining precise structural characterizations of expressive power. If two logics satisfy the compactness property, their comparison in terms of distinguishing power conincides with their comparison in terms of expressive power. This helps to obtain characterization of compact logics in terms of Lindstr\"{o}m-type theorems. However, compactness fails on non-elementary extension of FOL (e.g.\ modal $\mu$-calculus) or fragments of FOL restricted on finite models. We may look for different classes of models with different lifting method for such case to characterize the expressive power.

{\scriptsize
 \bibliographystyle{abbrv}

}

\newpage
\appendix
\noindent{\bf\Large Appendix}

\section{Proof of Theorem~\ref{Expr=ClassDist}}
\label{app-Expr=ClassDist}
\begin{proof}
\begin{description}
 \item[($1\Rightarrow 2$)] Suppose $L_2$ is at least as class-distinguishing as $L_1$, and let
$\mc{C}$ be a class of models definable by $\varphi_1\in L_1$. Then $\mc{C}=\{M\in\mathcal{M}\mid M\vDash\varphi_1\}$ and its complement $\overline{\mc{C}}=\{M\in\mathcal{M}\mid M\not\vDash\varphi_1\}$. Since $\varphi_1$ separates $\mc{C}$ and its complement, it follows from the assumption that there exists a separating formula $\varphi_2\in L_2$ for $\mc{C}$ and its complement. But then either $\mc{C}=\{M\in\mathcal{M}\mid M\vDash\varphi_2\}$ and $\overline{\mc{C}}=\{M\in\mathcal{M}\mid M\nvDash\varphi_2\}$ or vice versa. In any case $\mc{C}$ is definable in $L_2$.

 \item[($2\Rightarrow 1$)] Suppose $L_2$ is at least as expressive as $L_1$, and let $\Ci$, $\Cii$ be classes of models in $\mathcal{M}$. Suppose there exists a separating formula $\varphi_1\in L_1$ for $\Ci$ and $\Cii$ such that $\Ci\subseteq \{M\mid M\vDash\varphi_1\}$ and $\Cii\subseteq \{M\mid M\nvDash\varphi_1\}$. Since $L_2$ is at least as expressive, there is a formula $\varphi_2$ in $L_2$ such that $$\{M\in\mathcal{M}\mid M\vDash\varphi_1\}=\{M\in\mathcal{M}\mid M\vDash\varphi_2\}.$$ But then obviously $\varphi_2$ is a separating formula for $\Ci$ and $\Cii$ in $L_2$.
 \end{description}
\qed\end{proof}
\section{Proof of Proposition~\ref{prop-diff}}
\label{app-prop-diff}
\begin{proof}
First note that for any $u\in U_\mc{C}$, if $u$ is not a principal ultrafilter generated by some $s\in S_{\mc{C}}$ then it is not a principal ultrafilter. Suppose not, then $u$ is generated by a $s'\not\in S_\mc{C}$. Since $\{s'\}\cap S_{\mc{C}}=\emptyset$, $u\not\in U_\mc{C},$ contradiction. Now we know that $U_{\mc{C}}|_{np}$ only contains non-principal ultrafilters. Since $S_\mc{C}$ is infinite, for any co-finite subset $X$ of $\biguplus\mc{C}$: $X\cap S_\mc{C}\not=\emptyset$. Thus we have: $$U_{\mc{C}}|_{np}=\{u\mid  \textrm{ $u$ contains all co-finite subsets of $\biguplus\mc{C}$ and $\forall X\in u: X\cap S_\mc{C}\not=\emptyset$ }\}$$
Suppose $u\in U'_\mc{C}.$ According to the definition, $u$ contains all cofinite subsets of $S_\mc{C}$. Now we show $u\in U_\mc{C}$ by checking whether all co-finite subsets of $\biguplus\mc{C}$ are in $u.$ Take a co-finite subset $X$ of $\biguplus\mc{C}$. It is not hard to see that $X\cap S_\mc{C}$ is co-finite in $S_\mc{C}.$ Thus $X\cap S_\mc{C}\in u$, then by closure property of ultrafilter $u$, $X\in u.$

If $\biguplus\mc{C}\backslash S_{\mc{C}}$ is infinite, then it is easy to see that $u^*\in U_\mc{C}$ where $$u^*=\{X\mid X \textrm{ is a co-finite subset of $\biguplus\mc{C}$}\}$$ however, $u^*\not\in U'_\mc{C}$ since it does not contain any co-finite subset of $S_\mc{C}.$
\qed 
\end{proof}

\section{Proof of Theorem~\ref{biggesthmclass}}
\label{app-biggesthmclass}
\begin{proof}
Suppose towards contradiction that there is a class of pointed Kripke models $\mc{C}\supset\mc{C}^m$ such that $\mc{C}$ enjoys Hennessy-Milner property. Then there is a pointed model ($M,s$ generated by $s$) in $\mc{C}$ such that $M,s$ is not m-saturated. Now we consider the ultrafilter extension of pointed model $M$: $\ue(M)$ (call the generated model $N,t$). From Theorem~\ref{uethem}, $M,s\equiv_{\ML}N,t$ and $N,t$ is m-saturated, thus $N,t\in\mc{C}$. Since $\mc{C}$ has Hennessy-Milner property, $M,s\bis N,t$. From Lemma~\ref{m-bis}, $M,s$ is m-saturated, contradiction.
\qed\end{proof}
\section{Proof of Lemma~\ref{m-bis}}
\label{app-m-bis}
\begin{proof}
Suppose $M,s\bis N,t$ and $N,t$ is m-saturated. Towards contradiction suppose that pointed model $M,s$ is not m-saturated. Then there is a point $s'$ such that there is a set of \ML\ formulas $\Sigma$ for which each finite subset is satisfiable at successors of $s'$, but $\Sigma$ itself is not satisfiable at any successor of $s'$. Let $t'$ be the bisimilar point of $s'$ in $N$. We claim $$\textrm{every finite subset of $\Sigma$ is satisfiable at the successors of $N,t'$}.$$
To prove the claim, observe that for any finite subset $\Delta$ of $\Sigma$, $\Diamond\bigwedge\Delta$ is satisfiable at $s'$. Since $M,s'\bis N,t'$ thus $\Diamond\bigwedge\Delta$ is satisfiable at $t'$, meaning that there is a successor of $t'$ such that $t'$ satisfies $\Delta$. Therefore, every finite subset of $\Sigma$ is satisfiable at the successors of $t'$.

Now based on the claim, and the fact that $N,t$ is m-saturated, we have that $\Sigma$ is satisfiable at some successor $r$ of $t'$. However, there can not be a bisimilar successor of $s'$ in $M$ to $r$ in $N$, otherwise $\Sigma$ is satisfiable in a successor of $s'$. Thus $M,s'\not \bis N,t'$, contradiction.\qed
\end{proof}
\section{Proof of Theorem~\ref{largestcollclassthm}}
\label{app-largestcollclassthm}
\begin{proof}
Suppose towards contradiction that there is a class $\mc{K}'\supseteq \mc{K}^m$ such that $(\circledast)$ holds for \ML. Then according to the definition of $\mc{K}^m$, there is a class $\mc{C}\in \mc{K}'$ such that either $\mc{C}|_m$ is not m-saturated or there is a model $M,s\in\mc{C}$ that is not equivalent to any m-saturated model in $\mc{C}$.

For the first case, there is a set of formulas $\Sigma$ such that it is not satisfiable in $\mc{C}|_m$ but each finite subset of $\Sigma$ is. Now consider the $\ue(\mc{C})$. It is easy to see that $\Sigma$ is satisfiable  in $\ue(\mc{C})$. Thus there is a m-saturated model $M,s$ in $\ue(\mc{C})$ such that $M,s\vDash\Sigma$. We now take the class $\mc{C}'=\{M,s\}$, it is clear that $\mc{C}'\in \mc{K}'$ since it is m-saturated. Then for any $\phi\in\Sigma$, $\mc{C}\nvDash\neg\phi$ since $\phi$ must be satisfiable somewhere in $\mc{C}$. Thus it is not hard to verify that $\mc{C}\asymp_\ML \mc{C}'$. Now according to our assumption that $\mc{K}'$ has property $(\circledast)$ for \ML, we know that there is a model $N,t$ in $\mc{C}$ which is bisimilar to $M,s$. However, then $N,t\vDash\Sigma$, contradiction.

For the latter case, since $M,s$ is not m-saturated and we know there is no modally equivalent m-saturated model of $M,s$ in $\mc{C}$, thus $\ue(M),\pi_s\not\in\mc{C}$. Now we consider $\mc{C}$ and $\mc{C}'=\{\ue(M),\pi_s\}$. It is easy to see that $\mc{C}\asymp_\ML\mc{C}'$. However, there is no model $N,t$ in $\mc{C}$ such that $N,t\bis\ue(M),\pi_s$, since otherwise according to Lemma~\ref{m-bis} and Theorem~\ref{ueequiv}, $N,t$ is m-saturated too and $N,t\equiv_\ML M,s.$
\qed\end{proof}
\end{document}